\documentclass[submission,copyright,creativecommons]{eptcs}
\providecommand{\event}{GandALF 2022}

\usepackage{iftex}

\title{Generating Tokenizers with Flat Automata}

\author{Hans de Nivelle 
\institute{School of Engineering and Digital Sciences, \\
           Nazarbayev University, Nur-Sultan City, Kazakhstan}
\email{hans.denivelle@nu.edu.kz}
\and 
Dina Muktubayeva
\institute{School of Engineering and Digital Sciences, \\
           Nazarbayev University, Nur-Sultan City, Kazakhstan}
\email{dina.muktubayeva@nu.edu.kz}}

\def\titlerunning{Generating Tokenizers with Flat Automata}
\def\authorrunning{H. de Nivelle and D. Muktubayeva} 

\begin{document}

\newtheorem{theorem}{Theorem}[section]
\newtheorem{example}[theorem]{Example}
\newtheorem{definition}[theorem]{Definition}
\newtheorem{lemma}[theorem]{Lemma}
\newtheorem{question}[theorem]{Question}
\newtheorem{corollary}[theorem]{Corollary}
\newtheorem{wish}[theorem]{Wish} 
\newtheorem{algorithm}[theorem]{Algorithm}
\newtheorem{exercise}[theorem]{Exercise}

\newcommand{\proof}{{\bf Proof \\}}
\newcommand{\proofend}{{\bf End of Proof \\}}

\newcommand{\pair}{{\rm pair}}
\newcommand{\pow}{{\rm pow}}

\newcommand{\hlijn}[1]{{ \( \overline{\hspace{#1}} \) }}

\newcommand{\een}{{\bf (1) }}
\newcommand{\twee}{{\bf (2) }}
\newcommand{\drie}{{\bf (3) }}
\newcommand{\vier}{{\bf (4) }}
\newcommand{\vijf}{{\bf (5) }}

\newcommand{\Null}{{\rm Null}} 
\newcommand{\First}{{\rm First}}
\newcommand{\Trans}{{\rm Trans}}

\maketitle

\begin{abstract}
   \noindent
   We introduce flat automata for automatic generation of tokenizers. 
   Flat automata are a simple representation of standard finite automata.
   Using the flat representation, automata can be easily constructed, combined and printed.
   Due to the use of border functions,
   flat automata are more compact than standard automata in the case where 
   intervals of characters are attached to transitions, and
   the standard algorithms on automata are simpler. 
   We give the standard algorithms for tokenizer construction with automata,
   namely construction using regular operations, determinization, and minimization.
   We prove their correctness. The algorithms work with intervals of characters,
   but are not more complicated than their counterparts on single characters. 
   It is easy to generate C++ code from the final deterministic automaton. 
   All procedures have been implemented in C++ and are publicly available. 
   The implementation has been used in applications and in teaching. 
\end{abstract}

\section{Introduction}

This paper is part of a project to obtain
a programming language for the implementation of logical algorithms.
Logic is special because its algorithms operate on trees that
have many different forms with different subtypes.
Algorithms need to distinguish the form of 
the tree, and take different actions dependent on this form. 
Intended applications of our language are
parts of theorem provers, or interactive verification systems. 
For the interested reader, we refer to 
(\cite{deNivelle2021b}). 
Part of this project is to obtain a working compiler. 
We have looked at existing tools for the generation
of the parser and the tokenizer, but none of them
fulfilled our needs. 
In particular, there was no bottom-up parser generation
tool available that supports modern $ C^{++}, $ 
and existing tokenizer generation tools
are not flexible enough.  
Existing tokenizer generators like LEX~(\cite{Lex1975sys})
and RE2C~(\cite{Re2c_sys}) generate the complete tokenizer,
which makes them unsuitable for our language. 
Our language uses Python-style indentation,
which requires that the tokenizer must generate
a token when the indentation level changes. 
Detecting a change of indentation level is quite
complicated, and it cannot be represented
by regular expressions. 
Lack of flexibility is a general problem, for example
$ C $ and $ C^{++} $ require that
the tokenizer has access to type information,
so that different tokens can be generated
for identifiers that represent a type name
or a template name. 
$ C^{++} $-11 allows use of \verb+>>+ to close
two template arguments at once
(for example in \verb|std::vector<std::pair<int,int>>|. 
In that case, \verb|>>| must be tokenized as
two separate \verb|>|.
In order to do this correctly,  
the tokenizer needs to know
if the parser is currently parsing a template argument. 

In order to obtain the required flexibility, we created a new implementation
that does not generate the complete tokenizer,
but which only 
cuts the input in small chunks, and classifies them by type.
We discuss details of our implementation 
in Section~\ref{Sect_concl_fut}.
In this paper, we concentrate on the representation
of finite automata used by our implementation. 
We use so-called border functions to represent
interval-based transitions. 
Instead of storing transitions of form
$ ( [ \sigma_1, \sigma_2 ], q ), $ 
(for characters between $ \sigma_1 $ and $ \sigma_2, $ go to 
state $ q $) we store only the points where
the behavior of the transition changes, i.e. the borders,
so instead we store $ ( \sigma_1, q ), ( \sigma+2, \# ), $ 
with $ \# $ denoting 'getting stuck'. 
For every character, the transition is determined by
the greatest border that is not greater than the character itself. 
When implementing transformations on
automata, border functions are much easier to deal with 
than intervals, because there is no need to distinguish
between the beginning and the end of an interval. 
All that needs to be looked at, are the borders.

In addition to the use of border functions,
we store the automata in an array (vector) using 
relative state references. This removes the need to represent 
automata as graphs, and the combinations that correspond to regular
operators become trivial. In most cases, the automata can be just
concatenated with the addition of a few $ \epsilon $ transitions. 

These two modifications result in a representation that is easy to
explain and implement, and whose automata are easy to read. 
This is useful both for teaching and for debugging. Big automata representing complete
tokenizers tend to be local, and our transformations preserve this locality. 

In general, our automaton representation is somewhat more
complicated than the standard representation, 
some of the correctness proofs become a bit more complicated, 
but the operations themselves are equally complicated.
The extra effort in defining the automata and proving
the operations correct pays off when the automata are applied:
The standard representation must be further adapted in order to make
it work in practice, while ours works without further adaptation. 
We have implemented flat automata in $ C^{++}, $ 
and the implementation is available from \cite{CompilerTools}.

In the next section, we will define alphabets and border functions.
In Section~\ref{Sect_acceptor}, 
we define acceptors, which are automata that can only
accept or reject. In Section~\ref{Sect_regexp} we explain
how to obtain acceptors by means of regular operations. 
We do not define regular expressions as separate entities,
instead we directly construct the automata. 
In Section~\ref{Sect_classifier},
we define classifiers, which are obtained by 
pairing acceptors with token names.
In Section~\ref{Sect_determinization} we adapt the standard 
determinization procedure to automata with border functions. The border
functions make it possible to keep the algorithm simple. 
In Section~\ref{Sect_minimization} we adapt state minimization to
our representation of automata. The algorithm can be
kept simple (as simple as for single characters) because
of the border functions. We use 
Hopcroft's algorithm (\cite{Hopcroft_min_1971}) with an adaptation 
of a filter from \cite{Efficient_depth2016}. 
In Section~\ref{Sect_concl_fut} we draw some conclusions,
and sketch possibilities for future work.  

\section{Preliminaries}

We will assume that alphabets are well-ordered sets. 
In the usual case where the alphabet is finite, it is sufficient
that there exists a total order on the alphabet. 
\begin{definition}
   \label{Def_alphabet} 
   An \emph{alphabet} is a pair $ ( \Sigma, < ), $ s.t.
   $ \Sigma $ is a non-empty set, and $ < $ is a well-order
   on $ \Sigma. $ 
   We define $ c_{\bot} = \min(\Sigma). $ 
   If $ \{ c' \in \Sigma | \ c < c' \} $ is non-empty, then we write 
   $ c^{+1} $ for $ \min \{ \, c' \in \Sigma \ | \ c < c' \, \}. $ 
\end{definition}
As far as we know, all alphabets in use, including ASCII
and Unicode (\cite{Unicode})
satisfy the requirements of Definition~\ref{Def_alphabet}
or can be adapted in such a way that they do.

Our aim is to define automata by means of intervals,
because in practice, many tokens 
(like for example numbers or identifiers) 
use intervals in their definitions.
Another advantage of use of intervals is that
is becomes possible to use large alphabets, like Unicode.

Dealing with intervals becomes easier if one removes
the distinction between start and end of interval. 
This can be done by storing only the points where
a new value starts, 
and creating a special value $ \# $ denoting 'not in any interval'. 
For example, when defining
identifiers, one may want to define a transition
to some state $ q, $ 
for $ \sigma \in \{ A, \ldots, Z \} \cup \{ a, \ldots, z \}, $ 
because all letters usually behave the same. 
This can be represented as 
$ \{ \, (A,q), \, ( Z+1, \# ), \,
        (a,q), \, ( z+1, \# ) \ \}. $ 
Here $ A, Z+1, a, z+1 $ are the borders where the behavior changes.
In order to determine the transition for a given symbol 
one needs to find the largest border that is not greater
than the symbol at hand. 
We will call a function, that is defined in this way,
a border function.

\begin{definition}
   \label{Def_borderfunc}
   Let $ ( \Sigma, < ) $ be an alphabet, let $ D $ be an arbitrary,
   non-empty set.
   A \emph{border function} $ \phi $ on $ (\Sigma,<) $ is a partial function
   from $ \Sigma $ to $ D, $ 
   defined for a finite subset of $ \Sigma, $ 
   but at least for $ c_{\bot}. $ 
   We write $ {\rm dom}(\phi) $ for the set of symbols 
   for which $ \phi $ is defined. 
   We call the set $ D $ the \emph{range} of $ \phi. $ 
   We will write border functions as sets of ordered pairs,
   whenever it is convenient. 
\end{definition}

\begin{definition}   
   For a given $ \sigma \in \Sigma, $ we first define 
   $ \sigma^{\leq} = \max \{ \, \sigma' \in {\rm dom}(\phi) \ | \ 
   \sigma' < \sigma \mbox{ or } \sigma' = \sigma \, \}. $ 
   After that, we define $ \phi^{\leq}(\sigma) = \phi( \sigma^{\leq} ). $ 
\end{definition}
It can be easily checked that 
$ \phi^{\leq}(\sigma) $ always exists and is uniquely defined,
because $ \phi $ is finite and has $ c_{\bot} $ in its domain. 

\begin{definition}
   \label{Def_bord_minimal} 
   Let $ \phi_1 $ and $ \phi_2 $ be two border functions
   on the same alphabet $ ( \Sigma, < ). $ 
   We say that $ \phi_1 $ and $ \phi_2 $ are \emph{equivalent}
   if for all 
   $ \sigma \in \Sigma, $ \ 
   $ \phi^{\leq}_1(\sigma) = \phi^{\leq}_2(\sigma). $ 

   \noindent
   We call $ \phi $ \emph{minimal} if there exists
   no equivalent $ \phi' \subset \phi. $ 
   We define \emph{the minimization of} $ \phi $ as the 
   $ \subseteq $-minimal border function that is equivalent
   to $ \phi. $ 
\end{definition}
Definition~\ref{Def_bord_minimal} uses the
fact that border functions can be viewed as sets of ordered pairs.
It can be easily checked that border functions can be minimized.
If $ \phi(\sigma_1) = \phi(\sigma_2), $ and there is no 
$ \sigma' $ with $ \sigma_1 < \sigma' < \sigma_2 $ in the domain
of $ \phi, $ then $ \phi(\sigma_2) $ can be removed from $ \phi. $ 

If for example both $ \phi(1) = \phi(3) = 4, $ 
and $ 2 $ is not in the domain of $ \phi, $ then removing $ 3 $ 
from the domain will 
not have effect on $ \phi^{\leq}. $ 
Assume that $ \Sigma = \{ -100, -99, \ldots, 99, 100 \}. $ 
Assume that 
$ \phi(-100) = -1, $ \ \ 
$ \phi(-4) = 3, \ \ \phi(2) = 8, $ and $ \phi(6) = 4, $ then
$ \phi^{\leq}(-100) = \phi^{\leq}(-3) = -1, $ \ \
$ \phi^{\leq}( -4) = \phi^{\leq}(1) = 3, $ \ 
$ \phi^{\leq}(2) = \phi^{\leq}(5) = 8, $ and 
$ \phi^{\leq}(6) = \phi^{\leq}(100) = 4. $  

\begin{definition}
   \label{Def_border_product}
   Let $ \phi_1 $ and $ \phi_2 $ be two border functions
   over the same alphabet $ (\Sigma,< ). $ 
   Let $ D_1 $ be the range of $ \phi_1 $ and let $ D_2 $ be the
   range of $ \phi_2. $ 
   We define the \emph{product} 
   $ \phi_1 \times \phi_2 $ as the border function 
   \[ \{ \ ( \ \sigma, ( \phi_1^{\leq}(\sigma), \phi_2^{\leq}(\sigma) ) \ ) 
        \ | \
        \sigma \in {\rm dom}(\phi_1) \cup {\rm dom}(\phi_2) \ \}. \]
   The range of $ \phi_1 \times \phi_2 $ is 
   $ D_1 \times D_2. $ 
\end{definition}

\begin{definition}
   \label{Def_border_appl}
   Let $ \phi $ be a border function over alphabet $ ( \Sigma, < ). $
   Let $ D $ be the range of $ \phi. $ 
   Let $ f $ be a function from $ D $ to some set $ D'. $  
   We define \emph{the application of} $ f $ \emph{on} $ \phi $
   as the minimization of
   \[ \{ \ ( \, \sigma, f( \phi( \sigma )) \, ) \ | \ 
                 \sigma \in {\rm dom}(\phi) \ \}. \]
   We will write $ f(\phi) $ for the application of $f$ on $ \Phi. $ 
\end{definition}

\section{Acceptors}
\label{Sect_acceptor}

We distinguish two types of automata which we call 
\emph{acceptor} and \emph{classifier}.
An acceptor can only accept or reject a word, while
a classifier is able to classify words. 
A complete tokenizer is a classifier, while single tokens are defined
by acceptors. A classifier is obtained by associating 
acceptors with token classes. 

Although acceptors can be directly defined in code
through initializers,  
it is inconvenient to do this, and we will construct them 
from regular expressions. We do not view regular expressions
as independently existing objects. Instead we view regular operators
as operators that work directly 
on acceptors. We have no data structure for regular expressions.

Acceptors are standard finite automata. We represent them in such a way that
the regular operations are easy to present and to implement.
In order to obtain this, we use a flat, linear representation 
which we will introduce shortly. 
In the literature, finite automata are traditionally represented by 
graphs whose vertices are states and whose
edges are labeled with symbols. 
(See for example \cite{Dragon2007,Sipser2013}).
This is implementable, but we believe
that our representation is simpler. There is no problem of memory
management, and printing automata is easy. When we print 
an automaton, the states are printed absolute instead of relative.

\begin{definition}
   \label{Def_acceptor}
   Let $ ( \Sigma, < ) $ be an alphabet. 
   An \emph{acceptor} $ {\cal A} $ over $ \Sigma $ is a finite sequence 
   \[ {\cal A} =
        (\Lambda_1, \phi_1), \ldots, (\Lambda_n, \phi_n) \ \ ( n \geq 0 ), \] 
   where each $ \Lambda_i \subseteq {\cal Z}, $ and each
   $ \phi_i $ is a border function from 
   $ \Sigma $ to $ {\cal Z} \cup \{ \# \}. $

   Each $ \Lambda_i $ denotes the set of epsilon transitions
   from state $ i, $ while 
   each $ \phi_i $ represents the set of non epsilon transitions from
   state $ i. $ 
 
   We call $ {\cal A} $ \emph{deterministic} if all $ \Lambda_i $ are empty.
   We often write $ \| {\cal A} \| $ instead of $ n $ for the
   size of $ {\cal A}. $ 
\end{definition} 

\noindent
We use the following conventions:
\begin{itemize}
\item
   $ \# $ means that no transition is possible. 
   Note that $ \phi(\sigma ) = \# $ should not be confused with
   '$\phi( \sigma )$ is undefined'.
   Due to the use of border functions, one has to explicitly state 
   that $ \phi(\sigma) $ has no transition, because otherwise
   $ \phi^{\leq}( \sigma ) $ would 'inherit' a transition
   from a $ \sigma' < \sigma. $ 
\item
   The initial state is always $ 1, $ and 
   the accepting state is always $ n+1, $ just outside of
   the acceptor.
\item
   State references in a $ \Lambda_i $ or $ \phi_i $ 
   are always relative to $ i. $ 
   That means that $ i $ itself is represented by $ 0, $ 
   $ i+1 $ is represented by $ 1, $ while  
   $ i-1 $ is represented by $ -1, $ etc. 
\item
   There are no transitions to states $ < 2 $ or states $ > n +1. $ 
\end{itemize}
Note that the last condition stipulates that the acceptor
cannot return to the first state during a run. 
Most of the constructions for combining acceptors 
become simpler with this condition. 
Forbidding transitions to the initial state of an automaton 
is common in the literature,
see for example \cite{Autotheory2006}.
An automaton with this property is usually called \emph{committing}.

In addition, it is usually required that there is exactly one accepting
state, and that there are no transitions going out of the accepting state.
These conditions are automatically fulfilled by our representation. 

\noindent
Acceptors can be non-deterministic, but all
non-determinism must be inside the $ \Lambda_i, $ 
i.e. in the form of $ \epsilon $-transitions. 
All acceptors constructed by the regular operations of 
Section~\ref{Sect_regexp} have this form. 
If one wants to represent a general non-deterministic automaton,
one has to remove transitions from the same state with overlapping intervals.
For example, a state can have transitions to different states
for the intervals $ [a, \ldots, z], $ and $ [ a, \ldots, d]. $ 
In this case, the original state can be split into two states
connected by an $ \epsilon $-transition. 
During this process,
the number of states and $ \epsilon $-transitions can increase, but it
will not become more than the total number of borders in the step
functions of the original automaton. 

We will now formally define when $ {\cal A} $ accepts a word $ w. $ 
\begin{definition}
   \label{Def_run} 
   Let $ {\cal A} $ be an acceptor over alphabet $ \Sigma. $ 
   We define a \emph{configuration} of $ {\cal A } $ as a pair $ ( z, w ), $ 
   with $ 1 \leq z \leq \| {\cal A} \| $ and $ w \in \Sigma^{*}. $ 

   \noindent
   We define the \emph{transition relation} $ \vdash $ between
   configurations as follows: 
   \begin{itemize}
   \item
      If $ j \in \Lambda_i $ and $ w \in \Sigma^{*}, $ then 
      $ ( i, w ) \vdash ( i+j, w ). $ 
   \item
      If $ w \in \Sigma^{*}, \ \ \sigma \in \Sigma, $ and
      $ \phi_i^{\leq}(\sigma) = j $ with $ j \not = \#,  $ then
      $ (i, w ) \vdash ( i+j, w \sigma ), $ 
      where $ \phi^{\leq}_i(\sigma) $ is the border function of state $ i $ 
      applied on $ \sigma. $ 
   \end{itemize}

   \noindent
   We define $ \vdash^{i} $ and $ \vdash^{*} $ between configurations
   as usual. 
   
   \noindent
   We say that $ {\cal A} $ \emph{accepts}
   $ w \in \Sigma^{*} $ 
   if $ ( 1, \epsilon ) \vdash^{*} ( \, \| {\cal A} \| + 1, w \, ). $

   \noindent
   We write $ {\cal L}({\cal A}) $ for the language
   $ \{ w \in \Sigma^{*} \ | \ {\cal A} \mbox{ accepts } w \, \}. $ 
\end{definition} 
\begin{example}
   \label{Ex_ident}
   We give an acceptor that accepts standard identifiers
   (starting with a letter, followed by zero or more letters, digits,
    or underscores).
   The first column, which numbers the states, is not part
   of the automaton.
   \[ 
   \begin{array}{lll}
      1: \ \ & \{ \ \} \hspace*{0.3cm} & 
              \{ \, (c_{\bot},\# ), \, 
               (A,1), \, (Z^{+1}, \# ), \, ( a,1 ), \, ( z^{+1}, \# ) \, \} \\
      2: & \{ 1 \} & \{ \, (c_{\bot}, \# ), \, 
                      (0,0), \, (9^{+1}, \# ), \, 
                      ( A,0 ), \, ( Z^{+1}, \# ), \,
                      ( {\_},0 ), \, ( {\_}^{+1}, \# ), \, 
                      (a,0), \, (z^{+1}, \# ) \, \} \\ 
   \end{array}
   \]
   The initial state is $ 1. $ 
   From state $ 2, $ there is one epsilon transition to state $ 2+1 = 3, $ 
   which is the accepting state. 
   If $ \sigma \in \{ a, \ldots, z \} \cup \{ A, \ldots, Z \} \cup \{ \_ \}, $
   there is a transition from state $ 2 $ to state $ 2 + 0 = 2. $ 
\end{example}

\begin{example}
   \label{Ex_while}
   The following acceptor accepts the reserved word "while".
   The accepting state is $ 6. $ 
   \[
      \begin{array}{lll}
         1: \ \ \ & \{ \ \} \ \ \ \ \ & 
                   \{ \, (c_{\bot}, \#), (w, 1), ( w^{+1}, \# ) \, \} \\
         2: & \{ \ \} & 
                   \{ \, (c_{\bot},\#), (h, 1), ( h^{+1}, \#) \, \} \\
         3: & \{ \ \} & 
                   \{ \, (c_{\bot},\#), (i, 1), ( i^{+1}, \# ) \, \} \\
         4: & \{ \ \} & 
                   \{ \, (c_{\bot},\#), (l, 1), ( l^{+1}, \# ) \, \} \\
         5: & \{ \ \} & 
                   \{ \, (c_{\bot},\#), (e, 1), ( e^{+1}, \# ) \, \} \\
      \end{array}
   \]
\end{example}
It may seem from 
Examples~\ref{Ex_ident} and \ref{Ex_while} that acceptors can be easily
written by hand, but unfortunately that is not the case in general, 
because one needs to know
the order of the alphabet. One must remember that upper case
letters come before lower case letters in ASCII, and the relative
positions of special symbols. 
We initially thought that it would be doable, but it turned
out impossible to write non-trivial acceptors by hand. 
Despite this, automata are easily readable if one prints
the states in transitions as absolute, 
and uses the following printing
convention: In the transition function,
pairs of form $ ( \sigma, \# ) $ where $ \sigma $ 
is the successor of a symbol $ \tau, $ 
are printed in the form 
$ (\tau^{+1}, \# ). $ 
Without this convention, for example 
$ (A,0), (Z^{+1}, \# ) $ would be printed as
$ (A,0), ( \, [, \# ), $ which is a bit hard to read.

\section{Obtaining Acceptors by Regular Operations} 
\label{Sect_regexp}

As explained below Example~\ref{Ex_while},
writing down acceptors directly by hand is unpractical.
The standard approach in the literature 
and in existing
systems, is to obtain automata by means of regular expressions
(\cite{Dragon2007,Sipser2013}).
We follow this approach, but 
we will not view regular expressions as independently existing objects.  
Rather we 
define a set of regular operators on automata that construct 
acceptors at once. 

\begin{definition}
   \label{Def_acc_simple} 
   Let $ ( \Sigma, < ) $ be an alphabet. 
   In the current definition, we will construct 
   border functions with range $ \{ {\bf f}, {\bf t} \}. $ 
   We define 
   $ \phi_{\emptyset} = \{ ( \sigma_{\bot}, {\bf f} ) \}, $ and 
   $ \phi_{\Sigma} = \{ ( \sigma_{\bot}, {\bf t} ) \}. $ 
   We define 
   \[
      \phi_{\geq \sigma} = 
      \mbox{ if } ( \sigma = \sigma_{\bot} ) \mbox{ then } 
          \{ ( \sigma_{\bot}, {\bf t} ) \} \mbox{ else }
          \{ ( \sigma_{\bot}, {\bf f} ), ( \sigma, {\bf t} ) \}. 
   \]
   For $ \sigma \in \Sigma, $ let 
   $ C_{> \sigma} = \{ \sigma' \in \Sigma \ | \ \sigma' > \sigma \ \}, $ 
   the set of symbols greater than $ \sigma. $ 
   We define
   \[
      \phi_{\leq \sigma} =
      \mbox{ if } ( C_{> \sigma} = \emptyset ) \mbox{ then }
         \{ ( \sigma_{\bot}, {\bf t} ) \} \mbox{ else } 
         \{ ( \sigma_{\bot}, {\bf t} ), 
            ( {\rm min}(C_{> \sigma}), {\bf f} ) \} \\
   \]
   We define $ \phi_1 \cap \phi_2 = I( \phi_1 \times \phi_2 ), $ with
   $   I( \, ( d_1, d_2 ) \, ) = \mbox{ if } 
            ( d_1 = {\bf t} \mbox{ and } d_2 = {\bf t} ) \mbox{ then }
            {\bf t} \mbox{ else } {\bf f}, $
   and we define $ \neg \phi = N(\phi), $ with 
   $ 
      N(d) = \mbox{ if } ( d = {\bf t} )
         \mbox{ then } {\bf f} \mbox{ else } {\bf t}. $ 
   Other Boolean combinations, like $ \phi_1 \cup \phi_2, $ and
   $ \phi_1 \backslash \phi_2 $ can be defined analogously.
\end{definition}

\begin{definition}
   \label{Def_acc_from_borderfunc} 
   We define the following ways of constructing
   acceptors over $ \Sigma: $ 
   \begin{itemize}
   \item
      The acceptor $ {\cal A}_{\epsilon}, $ which accepts
      exactly the empty word, is defined as $ ( \, ). $ 
   \item
      Let $ f_{\#} $ be the function defined from 
      $ f_{\#}( {\bf f} ) = \#, $ and
      $ f_{\#}( {\bf t} ) = 1. $ 
      Then, if $ \phi $ is a border function with range 
      $ \{ {\bf f}, {\bf t} \}, $ 
      we define $ {\cal A}[ \phi ] $ as the acceptor 
      $ ( \, ( \{ \, \}, f_{\#}( \phi ) \, ). $ 
      (We are using Definition~\ref{Def_border_appl}.)
   \end{itemize}
\end{definition}
$ {\cal A}[ \phi ] $ accepts exactly the symbols (as words)
for which $ \phi^{\leq} $ returns $ {\bf t}. $ 
Using $ {\cal A}[ \phi ], $ it is easy to construct
acceptors for Boolean combinations of intervals. 
For example, an acceptor that accepts exactly letters can
be defined as $ {\cal A}[ \ ( \phi_{\geq a} \cap \phi_{\leq z} ) \cup
                            ( \phi_{\geq A} \cap \phi_{\leq Z} ) \ ]. $ 
An acceptor that accepts all letters except X can be defined as
$ {\cal A}[ \ \phi_{\Sigma} \cap \neg ( \phi_{\geq X} \cap \phi_{\leq X} ) \ ]. $
The acceptor that accepts nothing can be defined as 
$ {\cal A}_{\emptyset} = {\cal A}[ \phi_{\emptyset} ]. $ 

\begin{definition}
   \label{Def_operations}
   Let $ {\cal A} =  ( \Lambda_1, \Phi_1 ), \ldots, (\Lambda_n, \Phi_n) $ 
   and $ {\cal A}' = 
             ( \Lambda'_1, \Phi'_1 ), \ldots, (\Lambda'_n, \Phi'_{n'}) $
   be acceptors. We define the \emph{concatenation} 
   $ {\cal A} \circ {\cal A}' $ as $ 
     ( \Lambda_1, \Phi_1 ), \ldots, ( \Lambda_n, \Phi_n), 
     ( \Lambda'_1, \Phi'_1 ), \ldots, (\Lambda'_{n'}, \Phi'_{n'} ). $ 
\end{definition}
Operation $ \circ $ simply concatenates acceptors.

\begin{theorem}
   \label{Theorem_concatenation}
   Let $ {\cal A}_1 $ and $ {\cal A}_2 $ be acceptors. 
   $ {\cal L} ({\cal A}_1 \circ {\cal A}_2) = 
   \{ \, w_1 w_2 \ | \ w_1 \in {\cal L}({\cal A}_1) \; and \; 
                      w_2 \in {\cal L}({\cal A}_2) \, \} $.
\end{theorem}
\begin{proof}
Throughout the proof,
we define $ n_1 = \| {\cal A}_1 \| $ and $ n_2 = \| {\cal A}_2 \|. $ 

Let $w \in {\cal L}( {\cal A}_1 \circ {\cal A}_2 ). $ 
By definition, $ (1,\epsilon) \vdash^{*} ( n_1 + n_2 + 1, w ). $ 
There exists at least one prefix $ w' $ of $ w, $ s.t.
$ ( 1, \epsilon ) \vdash^{*} ( n',w' ) \vdash^{*} ( n_1 + n_2 + 1, w ) $ 
having $ n' > n_1 $ because $ w $ itself satisfies this condition.
Let $ w_1 $ be the smallest such prefix.
By the last condition of Definition~\ref{Def_acceptor}, 
$ n' $ must be equal to
$ n_1+1, $ hence $ w_1 \in {\cal L}( {\cal A}_1 ). $
Let $ w_2 $ be the rest of $ w, $ so we have $ w = w_1 w_2. $ 
Because $ ( n_1 + 1, w_1 ) \vdash^{*} ( n_1 + n_2 + 1, w ), $ 
it follows that $ ( n_1+1, \epsilon ) \vdash^{*} ( n_1 + n_2 + 1, w_2 ). $
Note that this sequence still uses $ {\cal A}_1 \circ {\cal A}_2. $ 
Since $ {\cal A}_2 $ has no transitions to states $ < 2, $ 
and all transitions originate from $ {\cal A}_2, $ the configurations
$ (n'',w'') $ in the sequence 
$ ( n_1 + 1, \epsilon ) \vdash^{*} ( n_1 + n_2 + 1, w_2 ) $ 
must have $ n'' \geq n_1 + 1 $. Since transitions are relative,
we have $ ( 1, \epsilon ) \vdash^{*} ( n_2 + 1, w_2 ) $
in $ {\cal A}_2. $ 

Now assume that $ w_1 \in {\cal L}( {\cal A}_1 ) $ and 
$ w_2 \in {\cal L}( {\cal A}_2 ) $. 
We have $ (1,\epsilon) \vdash^{*} (n_1,w_1) $ in $ {\cal A}_1, $ 
and $ (1,\epsilon \vdash^{*} (n_2,w_2) $ in $ {\cal A}_2. $
The second sequence can be easily modified into
$ ( n_1+1, \epsilon ) \vdash^{*} (n_1+n_2+1, w_2 ) $ in 
$ {\cal A}_1 \circ {\cal A}_2, $
which in turn can be modified into
$ (n_1+1,w_1) \vdash^{*} (n_1 +n_2+1, w_1 w_2 ) $ 
in $ {\cal A}_1 \circ {\cal A}_2. $ 
\end{proof}

\begin{definition}
   \label{Def_union}
   We first define an operation that adds $ \epsilon $ transitions
   to an acceptor.
   Let $ {\cal A} = ( \Lambda_1, \Phi_1 ), \ldots, $ $ (\Lambda_n, \Phi_n) $
   be an acceptor.
   We define $ {\cal A} \{ i \rightarrow^{\epsilon} j \} $ 
   as 
   $ ( \Lambda_1, \Phi_1 ), \ldots, 
     ( \Lambda_i \cup \{ j - i \}, \Phi_i ), \ldots, 
     ( \Lambda_n, \Phi_n ). $ 
   We add $ j-i $ instead of just $ j $ to $ \Lambda_i $ because transitions 
   are relative.

   \noindent
   The \emph{union} $ {\cal A}_1 \, | \, {\cal A}_2 $ of
   $ {\cal A}_1 $ and $ {\cal A}_2 $ is defined as 
   \[ ( \, {\cal A}_1 \circ {\cal A}_{\emptyset} 
           \circ
           {\cal A}_2 \, ) \{ \ 
            \ 1 \rightarrow^{\epsilon} \| {\cal A}_1 \| + 2, \ \ \allowbreak
            \| {\cal A}_1 \| + 1 \rightarrow^{\epsilon}  
            \| {\cal A}_1 \| + \| {\cal A}_2 \| + 2 \ \}. \]
\end{definition}
In this definition, we use $ {\cal A}_{\emptyset} $ as defined
below Definition~\ref{Def_acc_from_borderfunc}, namely 
$ {\cal A}_{\emptyset} = ( \{ \, \}, \{ ( c_{\bot}, \# ) \} ). $ 
We prove that union behaves as expected: 
\begin{theorem}
\label{Theorem_union}
   For every two acceptors $ {\cal A}_1$ and ${\cal A}_2, $ we have
   $ {\cal L} ({\cal A}_1 \, | \, {\cal A}_2) = 
   {\cal L} ({\cal A}_1) \cup {\cal L}({\cal A}_2)  $.
\end{theorem}
\begin{proof}
   As before, we use $ n_1 = \| {\cal A}_1 \| $ and $ n_2 = \| {\cal A}_2 \|. $
   Assume that $ w \in {\cal L}( {\cal A}_1 | {\cal A}_2 ). $
   By definition, $ (1,\epsilon) \vdash^{*} ( n_1 + n_2 + 2, w ) $.
   If state $ n_1 + 2 $ does not occur in this sequence, it must be the case 
   that the state $ n_1 + 1 $ occurs in the sequence, because the accepting 
   state is reachable only from $ n_1 + 1 $ or from states
   $ \geq n_1 + 2. $. This implies that $ w \in {\cal L}( {\cal A}_1) $. 
   Similarly, if state $ n_1 + 2 $ occurs in the sequence, then
   we note that $ n_1 + 2 $ originates from the initial state of
   $ {\cal A}_2. $ It follows that $ w \in {\cal L}( {\cal A}_2 ) $. 
   As a consequence, we have 
   $ {\cal L}( {\cal A}_1 | {\cal A}_2 ) \subseteq 
      {\cal L}( {\cal A}_1 \cup {\cal A}_2 ) $.

   Now assume that $ w \in {\cal L}( {\cal A}_1 ) \cup {\cal L}( {\cal A}_2 ) $.   
   If $ w \in {\cal L}( {\cal A}_1 ) $, we have 
   $ (1,\epsilon) \vdash^{*} ( n_1 + 1, w ) $ in $ {\cal A}_1 $. 
   In $ {\cal A}_1 | {\cal A}_2 $, this sequence can be extended to
   $ (1,\epsilon) \vdash^{*} ( n_1 + 1, w ) \vdash ( n_1 + n_2 + 2, w ) $. 

   If $ w \in {\cal L}( {\cal A}_2 ) $, we have 
   $ (1,\epsilon) \vdash^{*} ( n_1 + 1, w ) $ in $ {\cal A}_2 $. 
   In $ {\cal A}_1 | {\cal A}_2 $, this sequence becomes 
   $ (1,\epsilon) \vdash ( n_1 + 2, \epsilon ) \vdash^{*} ( n_1 + n_2 + 2, w ) $. 
   This implies that 
   $ {\cal L}( {\cal A}_1 \cup {\cal A}_2 ) \subseteq {\cal L}( {\cal A}_1 | {\cal A}_2 ) $. 
\end{proof}

\begin{definition}
   \label{Def_kleene_star}
   The Kleene star $ {\cal A}^{*} $ of $ {\cal A} $ is defined as
      \[ ( \, {\cal A}_{\emptyset} \circ 
              {\cal A} \circ {\cal A}_{\emptyset} \, ) 
          \{ \, 1 \rightarrow^{\epsilon} 2, \ 
              2 \rightarrow^{\epsilon} \| {\cal A} \| + 3, \
            \| {\cal A} \| + 2 \rightarrow^{\epsilon} 2 \, \}. \]
\end{definition}
\begin{theorem}
\label{Theorem_kleene_star}
   For every acceptor $ {\cal A}, $ the following holds: 
   \[ w \in {\cal L}({\cal A}^{*}) \mbox{ iff there exist }
   w_1, \ldots, w_k \ ( k \geq 0 ), \mbox{ s.t. }
      w = w_1 w_2 \cdots w_k \mbox{ and each } w_i \in {\cal L}({\cal A}). \]
\end{theorem}
\begin{proof}
   In this proof, let $ n = \| {\cal A} \|. $ 
   First assume that $ ( 1, \epsilon ) \vdash^{*} ( n+3, w ). $ 
   By separating out the visits of state $ 2, $ we can
   write this sequence in the following form: 
   \[ ( 1, \epsilon) \vdash ( 2, \epsilon ) 
       \vdash^{*} (2,v_1) \vdash^{*} (2,v_2) \vdash^{*} \cdots
       \vdash^{*} (2,v_{k-1}) \vdash^{*} (2,v_k) \vdash^{*} (n+3,w), \]
   where each subsequence $ \vdash^{*} $ contains no visits to
   state $ 2. $ 
   For simplicity, set $ v_0 = \epsilon. $ 
   Then, for $ i $ with $ 1 \leq i \leq k, $ the word
   $ v_{i-1} $ is a prefix of $ v_i. $ 
   For $ 1 \leq i \leq k, $ define the difference $ w_i $ such that
   $ v_{i-1} \, w_i = v_i. $ 
   
   By construction of $ {\cal A}^{*}, $ 
   state $ 2 $ originates from the original acceptor
   $ {\cal A}. $ Hence $ (2,v_{i-1} ) \vdash^{*} ( 2,v_i ) $ implies
   that $ (2,v_{i-1}) \vdash^{*} (2+n, v_i ) \vdash (2,v_i ). $ 
   Since the sequence $ (2,v_{i-1}) \vdash^{*} (2+n, v_i ) $ must be
   completely within $ {\cal A}, $ it follows that
   $ w_i \in {\cal L}( {\cal A} ). $ 
   For the final sequence $ (2,v_k) \vdash^{*} (n+3,w), $
   it can be easily checked that the only transition to
   $ n+3 $ is an $ \epsilon $-transition from state $ 2. $ 
   Hence, we have $ (2,v_k) \vdash (n+3,w) $ and $ v_k = w. $ 
   Since we have $ w = w_1 \cdots w_k, $ this completes one direction
   of the proof. 

   For the other direction, assume we have 
   $ w_1, \ldots, w_k, $ s.t each $ w_i \in {\cal L}({\cal A}) $ 
   for some $ k \geq 0. $ 
   By definition, we have $ ( 1, \epsilon ) \vdash^{*} ( n + 1, w_i ) $ 
   in $ {\cal A}, $ 
   which implies that for every word $ v' \in \Sigma^{*}, $ we have
   $ ( 1, v' ) \vdash^{*} ( n + 1, v' w_i ) $ in $ {\cal A}. $ 
  
   \noindent 
   In $ {\cal A}^{*}, $ we have $ ( 2, v' ) \vdash^{*} ( n + 2, v' w_i ). $ 
   By combining and properly instantating the $ v', $ we obtain
   \[ (1,\epsilon) \vdash ( 2, \epsilon ) \vdash^{*} 
      (2, w_1 ) \vdash^{*} (2,w_1 w_2 ) \vdash^{*} \cdots \vdash^{*} 
      (2,w_1 w_2 \cdots w_k ) \vdash ( n + 3, w_1 w_2 \cdots w_k ), \]
   which completes the proof.  
\end{proof}

\noindent
At this point, we can define all other common regular operations.
For example $ {\cal A}^{+} $ can be defined as $ {\cal A} \circ {\cal A}^{*}, $
and $ {\cal A}^{?} $ can be defined as $ {\cal A} | {\cal A}_{\epsilon}. $ 
Since direct construction results in slightly smaller
acceptors, we still give the following definitions:
\begin{definition}
   \label{Def_two_more}
   Let $ {\cal A} $ be an acceptor. We define 
   the non-empty repetition $ {\cal A}^{+} $ as 
   \[ ( \, {\cal A}_{\emptyset} \circ {\cal A} \circ 
           {\cal A}_{\emptyset} \, )
       \{ \ 1 \rightarrow^{\epsilon} 2, \ 
         \| {\cal A} \| + 2 \rightarrow^{\epsilon} 2, \ 
         \| {\cal A} \| + 2 \rightarrow^{\epsilon} \| {\cal A} \| + 3 \ \}. \]
   We define the optional expression $ {\cal A}^{?} $ as 
       $ {\cal A} \{ \ 1 \rightarrow^{\epsilon} \| {\cal A} \| + 1 \ \}. $
\end{definition}
The construction of $ {\cal A}^{?} $ relies on the fact that
$ {\cal A} $ is committing. 
\begin{theorem}
\label{Theorem_two_more}
   For every acceptor $ {\cal A}, $ the following holds:
   \[ w \in {\cal L}({\cal A}^{+}) \mbox{ iff there exist }
   w_1, \ldots, w_k \ ( k \geq 1 ), \mbox{ s.t. }
      w = w_1 w_2 \cdots w_k \mbox{ and each } w_i \in {\cal L}({\cal A}). \]
   \[ {\cal L}( {\cal A}^{?} ) = {\cal L}( {\cal A} ) \cup \{ \epsilon \}. \] 
\end{theorem}

\noindent
Instead of the automaton in example~\ref{Ex_ident}, we can now write:
\[ {\cal A}[ \, ( \phi_{\geq a} \cap \phi_{\leq z} ) \cup
                ( \phi_{\geq A} \cap \phi_{\leq Z} ) \, ] 
   \circ 
   {\cal A}[ \, ( \phi_{\geq a} \cap \phi_{\leq z} ) \cup
                ( \phi_{\geq A} \cap \phi_{\leq Z} ) \cup
                ( \phi_{\geq 0} \cap \phi_{\leq 9} ) \cup
                ( \phi_{\geq \_ } \cap \phi_{\leq \_ } ) \, ]^{*}. \]

\section{Classifiers}
\label{Sect_classifier}

In order to obtain a complete tokenizer, it is not sufficient
to accept or reject a given input. Instead one must
classify input into different groups. We call an automata
that can classify a \emph{classifier}. 
Contrary to standard text books, like for example
\cite{Sipser2013}, we define determinization
and minimization on classifiers, not on acceptors. 
\begin{definition}
   \label{Def_classifier}
   Let $ ( \Sigma, < ) $ be an alphabet. 
   Let $ T $ be a non-empty set of \emph{token classes}. 
   A \emph{classifier over} $ \Sigma $ \mbox{into} $ T $ 
   is a non-empty, finite sequence
   \[ {\cal C} = \, ( \Lambda_1, \phi_1, t_1), \ldots,
      ( \Lambda_n, \phi_n, t_n) \ \ ( n \geq 1 ), \]
   where each $ \Lambda_i \subseteq {\cal Z}, $ 
   each $ \phi_i $ 
   is a border function from $ \Sigma $ to $ {\cal Z} \cup \{ \# \}, $ 
   and each $ t_i \in T. $ 
   We will often write $ \| {\cal C} \| $ for
   the size of $ {\cal C}. $ 
   We call $ {\cal C} $ deterministic if all
   $ \Lambda_i $ are empty.
\end{definition}
For representing transitions, we use the same conventions as
for acceptors, namely that transitions are stored relative, and 
$ \phi_i(\sigma) = \# $ means that no transition is possible. 
In contrast to acceptors, we allow transitions to state $ 1, $
and we forbid transitions to state $ n+1. $ 
Intuitively, a classifier is a non-deterministic automaton,
which looks for the longest run possible, and classifies
as $ t_i $ when it gets stuck in state $ i. $
We will make this more precise soon.

In order to obtain a classifier, we start with a trivial
classifier that classifies every input as error
(actually, this classifier defines what is an error),
and add pairs of acceptors and token classes. 

We always assume that state $ 1 $ defines the error class.
This is a reasonable choice, because no classifier can
classify $ \epsilon $ as a meaningful token. 

\begin{definition}
   Let $ T $ be a token class. 
   Let $ e,t \in T $ and let $ {\cal A} =
      ( \Lambda_1, \phi_1), \ldots, ( \Lambda_n, \phi_n ) $ 
   be an acceptor. 
   We define $ {\cal A}[e,t] $ as the classifier
   $ ( \Lambda_1, \phi_1, e ), \ldots, ( \Lambda_n, \phi_n, e ), \,
     ( \{ \ \}, \{ ( \sigma_{\bot}, \# ) \}, t ), $ 
   i.e. as the classifier that classifies words accepted
   by $ {\cal A} $ as $ t, $ and all other words as $ e. $ 

   \noindent
   Let $ e \in T. $ 
   We define $ {\cal C}_e = ( \{ \ \}, \{ ( \sigma_{\bot}, 0 ) \}, e ), $
   i.e. as the classifier that classifies every word as $ e. $ 

   \noindent
   For a classifier $ {\cal C} $ with first classification
   $ t_1, $ acceptor
   $ {\cal A}, $ and $ t \in T, $ we 
   define $ {\cal C}[ \ t : {\cal A} \ ] $ as 
   \[ {\cal C} \{ \ 1 \rightarrow^{\epsilon} \| {\cal C} \| + 1 \ \}
        \circ {\cal A}[t_1,t]. \]
   Here $ \circ $ denotes concatenation of acceptors. 
\end{definition}
The construction of $ {\cal C}[ \ t: {\cal A} \ ] $
appends $ {\cal A} $ to $ {\cal C} $ in such a way that 
words accepted by $ {\cal A} $ will be classified as $ t. $
Since acceptors accept by falling out of the automaton,
we need to add an additional state without outgoing
transitions, which will classify
words that are able to reach it as $ t. $ 
We also add an $ \epsilon $ transition from the first state
to the added acceptor. 
Words that cannot reach an accepting state of any of the
acceptors will be classified
as $ t_1, $ because
the classification of the first state is used as error classification. 

\begin{example}
   \label{Ex_classifier}
   Assume that we want to construct a classifier that
   classifies identifiers as I with the exception
   of `while', which should be classified as W. 
   Using the acceptors of Examples~\ref{Ex_ident} and \ref{Ex_while}, we can
   construct $ {\cal C}_E [ \ I : {\cal A}_{\rm id}, \ 
    W : {\cal A}_{\rm while} \ ] $ as 
   \[  
   \begin{array}{llll}  
      1: \ & \{ 1, 4 \} & \{ \ (c_{\bot}, 0 ) \ \} & E \\ 
      \\  
      2: & \emptyset &
            \{ \, (c_{\bot}, \# ), \,
             (A,1), \, (Z^{+1}, \# ), \, ( a,1 ), \, ( z^{+1}, \# ) \, \} 
               \hspace*{0.5cm} & E \\
      3: & \{ 1 \} & \{ \, (c_{\bot}, \# ), \,
                 (0,0), \, (9^{+1}, \# ), \, ( A,0 ), \, ( Z^{+1}, \# ), \\
         & & \hspace*{1.5cm} 
               ( {\_},0 ), \, ( {\_}^{+1}, \# ), \, (a,0), \, (z^{+1}, \# ) \, \}  
                & E \\
      4: & \emptyset & \{ \, (c_{\bot}, \# ) \, \} & I \\
      \\ 
      5: & \emptyset & 
            \{ \, (c_{\bot}, \# ),  (w, 1), ( w^{+1}, \# ) \, \} & E \\
      6: & \emptyset & 
            \{ \, (c_{\bot}, \# ),  (h, 1), ( h^{+1}, \# ) \, \} & E \\
      7: & \emptyset & 
            \{ \, (c_{\bot}, \# ),  (i, 1), ( i^{+1}, \# ) \, \} & E \\
      8: & \emptyset & 
            \{ \, (c_{\bot}, \# ),  (l, 1), ( l^{+1}, \# ) \, \} & E \\
      9: & \emptyset & 
            \{ \, (c_{\bot}, \# ),  (e, 1), ( e^{+1}, \# ) \, \} & E \\
      10: & \emptyset & 
            \{ \, (c_{\bot}, \# ) \, \} & W \\ 
   \end{array}
   \]
\end{example}
Without further restrictions,
the classifier above can classify
`while' either as I or as W. In order to avoid such ambiguity,
we always take the classification of the maximal 
(using $ < $ on natural numbers) reachable state that is not
an error state. 
In the current case, after reading 'while' the reachable states are
$ 1, 3, 4 $ and $ 10. $ 
Since $ 10 $ is the maximal state and its label is not $ t_1 = E, $ 
the classifier classifies as W.

Other solutions for solving ambiguity do not work well.
In particular using an order $ < $ on $ T $ is unpleasant. 
If $ T $ is an enumeration type, it is difficult to control how 
$ T $ is ordered. If $ T $ is a string type, its order is
determined by the lexicographic order, and it is tedious to override it. 

\noindent
Before we can make classification precise, we need to introduce
one technical condition. By default, the first state defines
the error state $ t_1. $ If from the first state
it is possible to reach a state $ i $ with $ t_i \not = t_1, $ 
we could possibly classify the word as non-error.
Whenever we encounter such a situation in real, it is due to
a mistake, mostly due to writing $ {\cal A}^{*} $ where
$ {\cal A}^{+} $ would have been required.
Hence, we will forbid such automata.

\begin{definition}
   A classifier $ {\cal C} $ is \emph{well-formed} if
   it does not allow a sequence
   $ ( 1, \epsilon ) \vdash^{*} ( i, \epsilon ) $ 
   with $ t_i \not = t_i. $ 
\end{definition}
The automaton in Example~\ref{Ex_classifier} is well-formed.
Changing $ t_2 $ into $ t_2 = I $ would make it ill-formed. 

\noindent
The following definition makes classification precise:
\begin{definition}
   \label{Def_run_classifier}
   For classifiers, we define configurations 
   as in Definition~\ref{Def_run}. 
   We also define $ \vdash $ and $ \vdash^{*} $ in the same way.
   
   We define \emph{classification}: 
   Classifying a word $ w \in \Sigma^{*} $ means obtaining a maximal prefix
   $ w' $ of $ w $ that is not classified as error ($t_1)$, 
   together with 
   the preferred classification of $ w'. $ 
   Let $ {\cal C} $ be a classifier, let $ w \in \Sigma^{*}. $ 
   Let $ w' $ be a maximal prefix of $ w, $ s.t.
   there exists a state $ i $ of $ {\cal C} $ with  
   $ ( 1, \epsilon) \vdash^{*} ( i, w' ) $ and $ t_i \not = t_1. $ 

   If no such state $ i $ exists, then the classification of $ w $ equals
   $ ( \epsilon, t_1 ). $ 

   If such a state exists, assume that $ i $ is the largest state 
   for which $ ( 1, \epsilon) \vdash^{*} ( i, w' ) $ and $ t_i \not = t_1. $ 
   In this case, the classification equals $ ( w', t_i ). $ 
\end{definition}

\section{Determinization}
\label{Sect_determinization}

\noindent
It is possible to run a non-deterministic classifier directly,
but it is inefficient in the long run
when many input words need to be classified. 
As with standard automata,
a non-deterministic classifier can be transformed
into an equivalent, deterministic
classifier. The construction is almost standard 
(See for example~\cite{Dragon2007,Autotheory2006,Sipser2013}),
but there are a few differences:
We perform the construction on classifiers instead
of acceptors, because that is what will be used in
applications, and we get generalization to character intervals for free,
because of the use of border functions. 
The advantage of border functions is that there is no need to distinguish
between starts and ends of intervals. The only 
points that need to be looked at are the borders. 
As a result the construction is 
only slightly more complicated than the standard approach, 
while at the same time working in practice without adaptation.
The following definition is completely standard: 
\begin{definition}
   \label{Def_eps_clos}
   Let $ {\cal C} $ be a classifier.
   Let $ S $ be a subset of its states. 
   We define \emph{the closure} 
   of $ S, $ written as $ {\rm CLOS}_{\cal C}(S) $ 
   as the smallest set of states $ S' $ with
   $ S \subseteq S', $ and 
   whenever $ i \in S' $ and $ j \in \Lambda_i, $ we have
   $ i + j \in S'. $  
\end{definition}
As said before, during determinization one only needs to consider the 
borders: 
\begin{definition}
   \label{Def_border}
   Let $ {\cal C} $ be a classifier
   defined over alphabet $ ( \Sigma, < ). $ 
   Let $ S $ be a non-empty
   set of states of $ {\cal C}. $
   We define 
   \[ {\rm BORD}_{\cal C}( S ) = 
      \{ \, \sigma \in \Sigma \ | \ \sigma \mbox{ is in the domain of a }
         \phi_i \mbox{ with } i \in S \, \}. \]
\end{definition}
$ {\rm BORD}_{\cal C}(S) $ is the set of symbols 
where the border function of one of the
states in $ S $ has a border. These are the points where
'something happens', and which have to be checked when constructing
the deterministic classifier.
In the classifier of Example~\ref{Ex_classifier}, we have
$ {\rm CLOS}_{\cal C}( \{ 1 \} ) = \{ 1,2,5 \} $ and
\[ {\rm BORD}_{\cal C}( \{ 1,2,5 \} ) = 
   \{ \, c_{\bot}, \, A, \, Z^{+1}, \, a, \, w, \, w^{+1}, \, z^{+1} \, \}. \] 

\noindent
Before we describe the determinization procedure,
we need a way of extracting classifications from 
sets of states:

\begin{definition}
   \label{Def_set_class}
   Let $ {\cal C} $ be a classifier. Let 
   $ S $ be a subset of its states. 
   We define $ {\rm CLASS}_{\cal C}(S) $ as follows:
   If for all $ i \in S, $ one has $ t_i = t_1, $ then 
   $ {\rm CLASS}_{\cal C}(S) = t_1. $ 
   Otherwise, let $ i $ be the maximal element
   in $ S $ for which $ t_i \not = t_1. $ 
   We define $ {\rm CLASS}_{\cal C}(S) = t_i. $ 
\end{definition}

\noindent
In example~\ref{Ex_classifier},
$ {\rm CLASS}_{\cal C}( \emptyset ) = 
  {\rm CLASS}_{\cal C}( \{ 1,2,3,5,6,7,8,9 \} ) = E, $ \ \ \ 
$ {\rm CLASS}_{\cal C}( \{ 4, 6, 7 \} ) = I, $ and
$ {\rm CLASS}_{\cal C}( \{ 3, 4, 10 \} ) = W. $  

Now we are ready to define the determinization procedure.
It constructs a deterministic classifier
$ {\cal C}_{\rm det} $ from $ {\cal C}. $
\begin{definition}
   \label{Def_determinization}
   The determinization procedure
   maintains a map $ H $ that maps subsets of 
   states of $ {\cal C} $ that we have discovered into 
   natural numbers. 
   It also maintains a map $ S_i $ that is the inverse of $ H, $ 
   so we always have $ S_{H(S)} = S. $ 
   \begin{enumerate}
   \item
      Start by setting $ H( \ {\rm CLOS}_{\cal C}( \{ 1 \} ) \ ) = 1, $ 
      and by setting $ S_1 = {\rm CLOS}_C( \, \{ 1 \} \, ). $ 
   \item
      Set $ {\cal C}_{\rm det} = ( \, ). $ 
   \item
      As long as $ \| {\cal C}_{\rm det} \| < \| H \|, $ 
      repeat the following steps: 
   \item
      Let $ i = \| {\cal C}_{\rm det} \| + 1. $ 
      Append $ ( \{ \, \}, \{ \, \}, 
                 {\rm CLASS}_{\cal C}( S_i ) \, ) $ to
      $ {\cal C}_{\rm det}. $ 
   \item
      For every $ \sigma \in {\rm BORD}_{\cal C}(S_i), $ do the following:
      \begin{itemize}
      \item
         Let  
         $ S' = \{ \, s + \phi^{\leq}_s( \sigma ) \ | \
               s \in S \mbox{ and } 
                      \phi^{\leq}_s( \sigma ) \not = \# \, \}. $ 
         ($ \phi_s $ is the border function of state $ s. $)
      \item
         If $ S' = \emptyset, $ then extend $ \phi_i $ 
         by setting 
         $ \phi_i(\sigma) = \#. $ Skip the remaining steps.
      \item
         Set $ S'' = {\rm CLOS}_{\cal C}( S' ). $  
      \item
         If $ S'' $ is not in the domain of $ H, $ 
         then add $ H( S'' ) = \| H \| + 1 $ to $ H, $ 
         and set $ S_{\| H \| + 1} = S''. $
      \item
         At this point, we are sure that $ H(S'') $ is defined. 
         Extend $ \phi_i $ by setting $ \phi_i( \sigma) = H( S'' ). $ 
      \end{itemize}
   \end{enumerate}
\end{definition}
As usual, $ H $ and $ S $ can be discarded when the construction
of $ {\cal C}_{\rm det} $ is complete.
It is easily checked that $ {\cal C} $ is deterministic,
because all its $ \Lambda_i $ are empty.

\begin{theorem}
   \label{Theorem_determinization}
   Let $ {\cal C}$ be a classifier that is well-formed,
   and $ {\cal C}_{\rm det} $ be the classifier 
   constructed from $ {\cal C}$ by using the determinization procedure 
   of Definition~\ref{Def_determinization}. 
   For every word $ w \in \Sigma^{*}, $
   if $ {\cal C} $ classifies $ w $ as $ ( w', t' ), $ and
   $ {\cal C}_{\rm det} $ classifies $ w $ as $ (w'',t''), $ then
   $ w' = w'' $ and $ t' = t''. $ 
\end{theorem}
\begin{proof}
   The proof is mostly standard,
   and we sketch only the points where it differs from the standard proof. 
   Because $ {\cal C} $ is well-formed, we 
   have $ t_1 = t_{{det},1}, $ which means that both classifiers
   will use the same token class as error class. 

   \noindent 
   For every word $ w \in \Sigma^{*}, $ define the 
   set $ R_w = \{ \, r \in \{ \, 1, \ldots, \| {\cal C} \| \, \} \ | \ 
      (1,\epsilon) \vdash^{*} ( w, r ) \, \}. $ 
   These are the set of states that classifier $ {\cal C} $ can
   reach while reading $ w. $ 

   \noindent
   Also define the relation $ \delta_{det}(w,i) $ as
   $ ( \epsilon, 1 ) \vdash^{*} ( w, i ). $
   (Classifier $ {\cal C}_{det} $ reaches state $ i $ while
   reading $ w. $) 
   
   \noindent
   It can be proven by induction, that 
   \begin{enumerate}
   \item
      if $ R_w \not = \emptyset, $ then 
      $ \delta_{det}(w,i) $ implies $ i = H(R_w). $ 
      If $ R_w = \emptyset, $ then there is no $ i, $ s.t.
      $ \delta_{det}(w,i). $ 
   \item
      if $ R_w \not = \emptyset, $ then
      $ \delta_{det}(w,i) $ implies $ t_{det,i} = {\rm CLASS}(R_w). $  
   \end{enumerate}

   \noindent
   Now we can look at the classification of an
   arbitrary word $ w \in \Sigma^{*}. $ 
   If for all prefixes $ w' $ of $ w, $ we have
   $ {\rm CLASS}(R_{w'}) = t_1, $ then $ {\cal C} $ 
   will classify $ w $ as $ ( \emptyset, t_1). $ 
   If for some prefix there exists an $ i', $ s.t.
   $ {\delta}_{\det}(w',i') $ holds, we have
   $ t_{{det},i'} = {\rm CLASS}(R_{w'}) = t_1 $ by {\bf (2)}, 
   so that $ {\rm CLASS}(R_{w'}) = t_{\det,1}. $ 
   It follows that $ {\cal C}_{det} $ also classifies
   $ w $ as $ ( \emptyset, t_1 ). $ 

   If there exists a prefix $ w' $ of $ w $ 
   for which $ {\rm CLASS}(R_{w'}) \not = t_1, $ then let
   $ w' $ be the largest such prefix.
   There exists exactly one $ i', $ s.t.
   $ {\delta}_{\det}(w',i') $ holds, and by {\bf (2)} again,
   we have 
   $ t_{det,i'} = {\rm CLASS}(R_{w'}), $ which is not
   equal to $ t_{det,1}. $ 
  
   Because $ w' $ was chosen maximal, it follows that 
   for all words $ w'' \not = w' $ s.t. $ w' $ is a prefix of $ w'' $ and
   $ w'' $ is a prefix of $ w, $ either we have
   $ R_{w''} = \emptyset $ or $ {\rm CLASS}(R_{w''}) = t_1. $ 
   In both cases, there is no $ i'', $ s.t.
   $ (1,\epsilon) \vdash (i'',w'') $ and $ t_{det,i''} $
   in classifier $ {\cal C}_{det}. $ 
   In the former case, no $ i'' $ exists exists at all, 
   and in the latter case, 
   $ {\delta}_{\det}(w'', i'' ) $ holds, and we
   have $ t_{{det},i''} = {\rm CLASS}(R_{w''}) = t_1. $

   As a consequence, both $ {\cal C} $ and $ {\cal C}_{det} $ will
   classify $ w $ as $ ( w', {\rm CLASS}(R_{w'}) \, ). $  
\end{proof}

\section{State Minimization} 
\label{Sect_minimization}

It is well-known that for every regular language
there exists a unique deterministic automaton with
minimal number of states (See \cite{Dragon2007} Section~3.9,
or \cite{Autotheory2006} Section~4.4.3). The minimal
automaton can be obtained in time $ O(n.\log(n)) $
from any deterministic automaton
by means of Hopcroft's algorithm (\cite{Hopcroft_min_1971}). 

Altough it probably has minimal impact on performance,
minimization has a suprising effect on the size of the classifier. 
It turns out that on 
classifiers obtained from realistic programming
languages, the number of states decreases by $ 30/40 \%. $

It is straightforward to adapt Hopcroft's algorithm to
classifiers. We sketch the implementation below. 
The algorithm takes a deterministic classifier $ {\cal C} $ as input,
and constructs the smallest (in terms of equivalence classes)
partition on the states of $ {\cal C}, $ s.t. 
$ i \equiv j $ implies $ t_i = t_j $ and 
for every $ \sigma \in \Sigma, $ \ 
$ i + \phi^{\leq}_i(\sigma) \equiv j + \phi^{\leq}_j(\sigma). $
(We are implicitly assuming that $ \# \equiv \# $ and $ \# \not = i. $)
Once one has the partition, the automaton can be minimized
by selecting one state from each partition.

\begin{definition}
   \label{Def_min_data}
   We use an array $ (P_1, \ldots, P_p) $ for storing the 
   current state partition.
   We have 
   $ \bigcup_{1 \leq i \leq p} P_i = \{ \, 1, \ldots, \| {\cal C} \| \, \} $ 
   and $ i \not = j \Rightarrow P_i \cap P_j = \emptyset. $ 

   In addition to the partition $ (P_1, \ldots, P_p), $ we use 
   an index map $ I $ 
   that maps states to their partition, i.e. for
   every state $ i \, ( 1 \leq i \leq \| {\cal C} \| ), $ 
   we have $ i \in P_{I_i}. $ 
   
   The initial partition is obtained from a function
   $ f $ with domain $ \{ 1, \ldots, \| {\cal C} \| \, \} $
   and arbitrary range. 
   States $ i $ and $ j $ are put in the same class iff $ f(i) = f(j). $  
\end{definition}

We tried two initialization strategies: The first strategy
is simply taking
$f(r) = t_r, $ which means that two states will be equivalent
if they have the same classification. 
The second is an adaptation of a heuristic in 
\cite{Efficient_depth2016} 
that takes paths to possible future classifications into account.
We discuss it in more detail shortly. 

Due to the use of border functions
instead of intervals, Hopcroft's algorithm needs only
minor adaptation for classifiers in our representation.
We give the algorithm: 
\begin{definition}
   \label{Def_hopcroft_algorithm}
   First create an array $ B $ of back transitions. 
   For every state $ i $ with $ 1 \leq i \leq \| {\cal C} \|, $  \ \
   $ B(i) $ is the set of states that have a transition into $ i, $ i.e. \\
   $ B(s) = \{ \, j \ | \ 1 \leq j \leq \| {\cal C} \| 
         \mbox{ s.t. there exists a } \sigma \in \Sigma^{*}, \mbox{ s.t. } 
         \phi_j(\sigma) \not = \# \mbox{ and }
         j + \phi_j(\sigma) = i \, \}. $ 

   \noindent
   Construct the initial partition $ P = ( P_1, \ldots, P_p) $ from
   the chosen initialization function $ f. $ 
   Initialize the index array $ I $ from $ P. $  
   Create a stack $ U = ( 1, \ldots, p ) $ of indices. 
   The variable name $ U $ stands for \emph{unchecked}. 
    
   \begin{enumerate}
   \item
      As long as $ U $ is non-empty, pop an
      element from $ U, $ call it $ u, $ and do the following: 
   \item
      Construct $ S = \bigcup_{i \in P_u} B(i). $
      This is the set of states that have a transition
      into a state $ i \in P_u. $
   \item
      For every $ \sigma \in {\rm BORD}_{\cal C}(S) $ do: \ 
      Construct $ F_{\sigma} = 
      \{ \, i \in S \mid \phi^{\leq}_i(\sigma) \not = \# \mbox{ and }
                            i + \phi^{\leq}_i(\sigma) \in P_u \, \}. $ 
      Refine $ P,I,U $ with $ F_{\sigma}. $ 
    
      ($ F_{\sigma} $ is the set of states whose 
      $ \sigma $-transition goes into a state in $ P_u $) 
   \end{enumerate}

   \noindent
   The \emph{refinement operation} is defined as follows:
   Assume that we want to refine $ P,I,U $ by a set of states $ F. $ 
   For every $ P_i, $ s.t. $ P_i \cap F \not = \emptyset $ and
   $ P_i \not \subseteq F, $ do the following:
   \begin{enumerate}
   \item
      \label{Def_refine_sim}
      Construct $ N = P_i \backslash F $ and replace
      $ P_i $ by $ P_i \cap F. $ 
   \item
      If this results in $ \| P_i \| < \| N \|, $ 
      then exhange $ N $ and $ P_i. $ 
   \item
      Append $ N $ to $ (P_1, \ldots, P_i, \ldots, P_p), $ 
      and assign $ I(i) = p+1, $ for $ i \in N. $ 
      Add $ (p+1) $ to $ U. $  
   \end{enumerate}
   The intuition of refinement is the fact that
   if some $ P_i $ partially lies inside $ F $ and partially
   outside $ F, $ then $ P_i $ needs to be split. 
   
\end{definition}

\noindent
When the final partition $ (P_1, \ldots, P_p) $ has been obtained,
it is trivial
to construct the quotient classifier 
$ {\cal Q} = {\cal C} / (P_1, \ldots, P_p). $ 

It is essential that $ 1 \in P_1 $ because 
Definition~\ref{Def_classifier} and Definition~\ref{Def_run_classifier} 
treat $ t_1 $ as the error state. 
This can be easily obtained by sorting 
$ (P_1, \ldots, P_p) $ by their minimal element before
constructing the quotient classifier. 
An additional advantage of sorting is that it improves readability, because
it preserves more of the structure of the original classifier.

Both \cite{AroundHopcroft2006} 
and \cite{OnHopCroft2009} agree that $ U $ should
be implemented as stack, as opposed to a queue. 

Although Hopcroft's algorithm is theoretically optimal, it can 
be improved by a preprocessing stage. In the early stage
of the algorithm, all states that classify as error will be in a single
equivalence class. This equivalence class is gradually refined into 
smaller classes dependent on possible computations originating
from these classes. 
Although the number of steps is limited by the number of states in the class,
it may still be costly because the initial class is big. 

The initial refinements can be removed by using 
a preprocessing stage.
In \cite{Efficient_depth2016}, a filter 
for simple, deterministic automata is proposed that marks states with the 
shortest distance towards an accepting state. This can be done
in linear time. In order to adapt this approach
to classifiers, one has to include the accepted token in the markings.

\begin{definition}
   Let $ {\cal C} $ be a classifier from alphabet 
   $ (\Sigma, < ) $ into token set $ T. $ 
   A \emph{reachability function} $ \rho $ is a total 
   function from 
   $ \{ 1, \ldots, \| {\cal C} \} $ to partial
   functions from $ T $ to $ {\cal N}. $
\end{definition}
Intuitively, $ \rho(i)(t) = n $ means that 
there exists a path of length $ n $ from $ i $ to a state 
$ j $ with $ t_j = t. $ 

Although theoretically, the total size of $ \rho $ could be 
quadratic in the size of $ {\cal C}, $ 
in all cases that we encountered, all states except
for the initial state, can reach only a few token classes. 

Our goal is to compute the optimal reachability function 
and use it to initialize the first partition. 
This can be done with Dijkstra's algorithm.

\begin{definition}
   \label{Def_comp_reachability} 
   Start by setting
   $ \rho(i) = \{ \, (t_i,0) \, \}, $ for every state $ i $ 
   that has $ t_i \not = t_1. $ 
   (Every state can reach its own classification in $ 0 $ steps.) 
   Set $ \rho(i) = \{ \, \} $ for the remaining states 
   (that classify as error).

   \noindent
   Create a stack $ U = ( \, 1, \ldots, \| {\cal C} \| \, ) $ 
   of unchecked states. 
   \begin{itemize}   
   \item
      While $ U $ is not empty, pick and remove a state
      from $ U, $ call it $ u, $ and do the following:
   \item
      For every $ i \in B(u), $ for every
      $ ( t, n ) \in \rho(u) $ do the following:
      If $ \rho(i)(t) $ is undefined, insert $ ( t,n+1 ) $ to
      $ \rho(i). $ 
      Otherwise, if $ \rho(i)(t) = n', $ set 
      $ \rho(i)(t) = \min( n', n + 1 ). $ 

      \noindent
      If this results in a change of $ \rho(i), $ then
      add $ i $ to $ U. $ 
   \end{itemize}
\end{definition}
Using $ \rho $ to initialize the partition in 
Definition~\ref{Def_hopcroft_algorithm} works well in practice. 
In most cases, the first partition is also the final partition.
We end the section with an example of a reachability
function for a simple classifier that classifies
identifiers and the reserved word 'for': 

\begin{example}
   \label{Ex_reachability_while_for}
   Consider the following deterministic classifier that classifies
   identifiers (for simplicity only lower case and digits),
   and the reserved word 'for': 
   \[
      \begin{array}{llll}
         1: & \emptyset & \{ \, ( c_{\bot}, \# ), \, ( a, 1 ), \,  ( f, 2 ), \, ( g,1 ), \, ( z^{+1}, \# ) \, \} & E \\ 
         2: & \emptyset & \{ \, ( c_{\bot}, \# ), \, ( 0, 2 ), \,  ( 9^{+1}, \# ), \, ( a, 3 ), \, ( z^{+1}, \# ) \, \} & I \\
         3: & \emptyset & \{ \, ( c_{\bot}, \# ), \, ( 0, 1 ), \,  ( 9^{+1}, \# ), \, ( a, 2 ), \, ( o, 3 ), \, ( p, 2 ), \, ( z^{+1}, \# ) \, \} & I \\
         4: & \emptyset & \{ \, ( c_{\bot}, \# ), \, ( 0, 0 ), \,  ( 9^{+1}, \# ), \, ( a, 1 ), \, ( z^{+1}, \# ) \, \} & I \\ 
         5: & \emptyset & \{ \, ( c_{\bot}, \# ), \, ( 0, -1 ), \, ( 9^{+1}, \# ), \, ( a, 0 ), \, ( z^{+1}, \# ) \, \} & I \\
         6: & \emptyset & \{ \, ( c_{\bot}, \# ), \, ( 0, -2 ), \, ( 9^{+1}, \# ), \, ( a, -1 ), \, ( r,1 ), \, ( s, -1 ), \, ( z^{+1}, \# ) \, \} & I \\
         7: & \emptyset & \{ \, ( c_{\bot}, \# ), \, ( 0, -3 ), \, ( 9^{+1}, \# ), \, ( a, -2 ), \, ( z^{+1}, \# ) \, \} & F \\
      \end{array} 
   \]
   This classifier was constructed by the determinization procedure. 
   If one initializes the partition with $ f(i) = t_i, $ 
   the initial partition will be $ 
   ( \{ 1 \}, \{ 2,3,4,5,6 \}, \{ 7 \} ). $ 
   The optimal reachability function has 
   \[ \begin{array}{ll}
         \rho(1) = & \{ \, \, (I,1), \, (F,3) \, \} \\ 
         \rho(2) = \rho(4) = \rho(5) = & \{ (I,0) \, \} \\
         \rho(3) = & \{ \, (I,1), \, (F,2) \, \} \\
         \rho(6) = & \{ \, (I,1), \, (F,1) \, \} \\
         \rho(7) = & \{ \, (I,1), \, (F,0) \, \} \\
       \end{array}
   \] 
   The minimal classifier has $ 5 $ states, so the initial
   partition based on $ \rho $ is already the final partition. 
\end{example}

\section{Conclusions and Future Work}
\label{Sect_concl_fut}

We have introduced a way of representing finite
automata which uses relative state references and border functions. 
Border functions make it possible to concisely represent
interval-based transition functions. 
Our representation is more complicated than the
standard representation in text books (like \cite{Dragon2007,Sipser2013})
and the proofs are slightly harder, but the algorithms are not, 
and the representation can be used in practice without further adaptation. 
We have implemented our representation and used it
in practice. We gave a presentation about it,
together with our parser generation tool, at
the $ C^{++} $~Now conference.
The implementation is available from \cite{CompilerTools}. 

On the practical level, we make the threshold for using
our automated tools as low as possible.
In the simplest case,
one compiles the library, defines a classifier in code by means
of regular expressions, and calls a default function for 
classification. Constructing classifiers in code has the
advantage that the user does not need to learn a dedicated
syntax, and that construction of classifiers has full flexibility. 

Our implementation does not construct a complete tokenizer. 
This is important, because in our experience
this is the obstacle that stopped us from using an existing
tokenizer generator tool. 
There is always something in the language 
that cannot be handled by an automatically generated tokenizer. 
Therefore, in our implementation,
we automated only the classification process, 
and leave all remaining implementation to the user. 
In practice, not much additional code needs to be written. 
If one needs efficiency, one can create an executable classifier in 
$ C^{++}. $
Both the default classifier and the $ C^{++} $ classifier can be 
compiled with any input source which satisfies a small set
of interface requirements. 

In the future, we plan to look into full Boolean operations
(extend regular expressions with intersection and negation), 
or more advanced
matching techniques, as specified by POSIX.

The final point that needs consideration is the use of
compile time computation. 
Compile time computation was
introduced in $ C^{++} $-11 with the aim of allowing more general 
functions in declarations, 
primarily for the computation of the size of a fixed-size array. 
Since then, the restrictions on compile time computation
have gradually been relaxed, 
and nowadays, it is possible to convert a regular expression
represented as an array of characters into a table-based DFA
at compile-time. 
This was implemented in the CTRE library (\cite{CTRE2019}). 
We did not try to make our implementation suitable for
compile time computation, because it 
would result in reduced expressivity in the code that
constructs the acceptors. In addition, the experiments 
with 
RE2C imply that directly coded automata are an order of magnitude
faster than table-based automata (\cite{mythsfacts1998}). 

\section{Acknowledgements}

This work gained from comments by Witold Charatonik and
Cl{\'a}udia Nalon.
We thank Nazarbayev University for
supporting this research through the Faculty Development Competitive
Research Grant Program (FDCRGP) number 021220FD1651.

\bibliographystyle{eptcs}
\bibliography{flatautomata}

\end{document}